\documentclass[%
reprint,
amsmath,amssymb,
aps,
amsthm,
pra,
]{revtex4-2}

\usepackage{ifthen}
\newboolean{arxiv}
\setboolean{arxiv}{true}
\ifthenelse{\boolean{arxiv}}{\pdfoutput=1}{}

\usepackage{theorem} 
\usepackage{ascmac}
\usepackage{bbm}
\usepackage{bm} 
\usepackage{braket}
\usepackage{dcolumn} 
\usepackage{enumerate}
\usepackage{enumitem}
\usepackage{framed}
\usepackage{graphicx} 
\ifthenelse{\boolean{arxiv}}{\usepackage{hyperref}}{\usepackage[dvipdfmx]{hyperref}}
\usepackage{longtable}
\usepackage{mathtools}
\usepackage{multirow}
\ifthenelse{\boolean{arxiv}}{}{\usepackage{pxjahyper}}
\usepackage{txfonts} 
\usepackage{xparse}

\usepackage{color}
\usepackage[most]{tcolorbox}

\ifthenelse{\boolean{arxiv}}{}{\mathtoolsset{showonlyrefs}}

\hypersetup{
  bookmarksopen,
  bookmarksopenlevel=4,
  bookmarksnumbered
}

\usepackage{settings_common_paper}
\usepackage{settings}

\newcommand{\titlename}{Improving quantum channel discrimination with resourceful states}

\begin{document}

\preprint{APS/123-QED}

\title{\titlename}

\affiliation{%
 Quantum Information Science Research Center, Quantum ICT Research Institute, Tamagawa University,
 Machida, Tokyo 194-8610, Japan
}%

\author{Kenji Nakahira}
\affiliation{%
 Quantum Information Science Research Center, Quantum ICT Research Institute, Tamagawa University,
 Machida, Tokyo 194-8610, Japan
}%

\date{\today}

\begin{abstract}
 One of the key issues in quantum discrimination problems is understanding
 the extent of the advantages in discrimination performance
 when using resource states compared to resourceless states.
 We show that in any resource theory of states, which may not be convex,
 the extent to which the maximum average success probability can be improved
 in quantum channel discrimination problems without using auxiliary systems
 can be precisely quantified by the robustness measure.
 This result offers an intuitive operational meaning of the robustness measure
 in any convex resource theory.
 Furthermore, we demonstrate that the robustness measure can also quantify
 the improvement in channel discrimination problems that use auxiliary systems.
 Using these findings, resources can be fully characterized to achieve
 higher success probabilities than any state without the given resource
 in channel discrimination problems.
\end{abstract}

\maketitle



\section{Introduction}

Discrimination problems are fundamental in considering the operational properties
of quantum theory.
Among these, quantum channel discrimination stands out as a critical task,
focused on the precise identification of quantum channels by transmitting an input state
through them and measuring the resulting output states.
This capability is essential for a wide range of applications,
including quantum communication \cite{Gis-The-2007}, quantum sensing \cite{Deg-Rei-Cap-2017},
and quantum cryptography \cite{Gis-Rib-Tit-Zbi-2002}.
The ongoing research in this field
\cite{Aci-2001,Sac-2005,Sac-2005-EB,Li-Qiu-2008,Mat-Pia-Wat-2010,
Chi-2012,Sed-Zim-2014,Pir-Lup-2017,Puc-Paw-Kra-Kuk-2018,Pir-Lau-Lup-Per-2019,Nak-Kat-2021}
is pivotal in enhancing the accuracy of channel discrimination,
thereby significantly boosting the efficiency and security of quantum protocols.

In channel discrimination problems, appropriately selecting the state to input
into the channels directly impacts the effectiveness of the discrimination protocol.
For example, it is known that using any entangled state as the input state can improve
channel discrimination performance compared to using any separable state \cite{Pia-Wat-2009}.
While entanglement is a typical resource in quantum information processing,
various other types of quantum resources are also explored, such as
coherence\cite{Bau-Cra-Ple-2014,Win-Yan-2016,Str-Ade-Ple-2017},
superposition\cite{The-Kil-Egl-Ple-2017},
asymmetry\cite{Gou-Spe-2008,Mar-Spe-2016},
magic\cite{Vei-Mou-Got-Eme-2014,How-Cam-2017},
non-Gaussianity\cite{Gen-Par-Ban-2008,Tak-Zhu-2018,Alb-Gen-Par-Fer-2018},
non-Markovianity\cite{Wak-2017,Bha-Bha-Maj-2020},
quantum thermodynamics\cite{Bra-Hor-Opp-Ren-2013,Gou-Mul-Nar-Spe-2015},
and imaginarity\cite{Wu-Kon-Ran-Sca-2021,Xue-Guo-Li-Ye-2021}.
These resources can sometimes be treated in a unified manner as quantum resource theories.
Note that similar research has been conducted on the relationship between resourceful
measurements \cite{Car-Hei-Toi-2018,Car-Hei-Toi-2019,Skr-Sup-Cav-2019,Uol-Kra-Sha-Yu-2019,Skr-Lin-2019}
and channels \cite{Li-Bu-Liu-2020,Tak-Wan-Hay-2020,Hsi-Str-Wu-Ku-2025} and
their performance in discrimination problems.
Recently, Takagi \textit{et al.} \cite{Tak-Reg-Bu-Liu-2019} have shown that in any convex resource theory,
every resource state is useful for some channel discrimination task
compared to any resourceless state, called free state.
Continuous analysis of resource states leads to the refinement of strategies to achieve
optimal performance in various quantum systems and contributes to extensive discussions
on the strategic use of quantum resources.

A natural question arises as to how much the performance can be improved in channel discrimination tasks
using a given resource state compared to free states.
For this problem, the ratio of the maximum average success probability
in channel discrimination problems using a given state to that using the best free state
has been a focal point of research.
Studies have been conducted to quantify this ratio, and in certain resource theories
(e.g., entanglement \cite{Bae-Chr-Pia-2019},
steerability \cite{Pia-Wat-2015}, coherence \cite{Nap-Bro-Cia-Pia-2016,Bu-Sin-Fei-Pat-2017},
and asymmetry \cite{Pia-Cia-Bro-Nap-2016}), it has been shown that the supremum of this ratio,
which we call the discrimination power, can be precisely quantified by a measure
known as the (generalized) robustness (see also \cite{Tak-Reg-Bu-Liu-2019,Tak-Reg-2019}).
However, these results are limited to specific types resource theories,
and it was unclear whether similar results would hold for other resource theories.
It is also well known that in channel discrimination problems, the maximum average success probability
can generally be improved by using auxiliary systems (e.g., \cite{Wat-2018}).
Therefore, the discrimination power of channels with auxiliary systems
has also been focused.
Nevertheless, it is generally difficult to determine the discrimination power compared to the case
without auxiliary systems, and it has only been quantified in some special cases
\cite{Bae-Chr-Pia-2019,Pia-Wat-2015}.

In this paper, we address the question of whether the discrimination power
can be quantified by the robustness measure in any resource theory of states
in finite-dimensional systems, which has been mentioned as an important outstanding open question
in Ref.~\cite{Tak-Reg-Bu-Liu-2019}.
Our results provide an intuitive operational meaning of the robustness measure
in convex resource theories.
Additionally, since robustness can be formulated as the optimal value of a convex programming problem,
our results are useful for analytically or numerically determining the discrimination power.
First, we show that in channel discrimination problems without auxiliary systems
for any convex resource theory, the discrimination power of any resource state can be quantified by
the robustness.
This implies that the robustness measure can be interpreted as the discrimination power,
which is an intuitively clear operational value.
We also show that this result can be easily extended to nonconvex resource theories.
Furthermore, we demonstrate that in channel discrimination problems using auxiliary systems,
the discrimination power of any resource state can similarly be quantified by the robustness.
This result provides a necessary and sufficient condition for the maximum average success probability
of a given resource state to be higher than that of any free state.

\section{Discrimination power}

Let us consider the problem of distinguishing a collection of channels $\SetL_{n=1}^N$
from system $A$ to system $B$.
The prior probabilities are given by $\SetP_{n=1}^N$.
When a state $\rho$ is input into system $A$ and a measurement $\Pi \coloneqq \SetPi_{n=1}^N$
is performed on the output state in system $B$, the average success probability is given by
\begin{alignat}{1}
 \PS(\rho,\Pi;\SetPL) &\coloneqq \sum_{n=1}^N p_n \Tr[\Pi_n \cdot \Lambda_n(\rho)].
 \label{eq:PS_rho_Pi}
\end{alignat}
When distinguishing $\SetL_{n=1}^N$ with just one use of the channels,
any discrimination strategy can be represented using such a pair of $\rho$ and $\Pi$.
Thus, the maximum average success probability when $\rho$ is fixed is given by
\begin{alignat}{1}
 \PS(\rho;\SetPL) &\coloneqq \max_\Pi \PS(\rho,\Pi;\SetPL).
 \label{eq:PS_rho}
\end{alignat}

Let $\Den_A$ denote the set of all density matrices of system $A$.
We consider a subset $\Free$ of $\Den_A$,
the elements of which are regarded as free states.
For any state $\rho$ of $A$, the supremum ratio of the maximum average success probability
using $\rho$ to that using free states is defined as
\begin{alignat}{1}
 G(\rho,\Free) &\coloneqq \sup_{\SetPL}
 \frac{\PS(\rho;\SetPL)}{\sup_{\w \in \Free} \PS(\w;\SetPL)}
 \label{eq:G}
\end{alignat}
and is referred to as the discrimination power of $\rho$ with respect to $\Free$.
For some $\Free$, it has been shown that $G(\rho,\Free) = 1 + R_\Free(\rho)$ holds
\cite{Bae-Chr-Pia-2019,Nap-Bro-Cia-Pia-2016,Bu-Sin-Fei-Pat-2017,Pia-Cia-Bro-Nap-2016},
where $R_\Free(\rho)$ is the robustness measure, defined as
\begin{alignat}{1}
 R_\Free(\rho) &\coloneqq \inf \left\{ \lambda \ge 0 \relmid \frac{\rho + \lambda \tau}{1 + \lambda}
 \mathrel{\eqqcolon} \w \in \Free, ~ \tau \in \Den_A \right\}.
 \label{eq:R}
\end{alignat}
$R_\Free(\rho)$ can be interpreted as the least coefficient needed for the mixture of $\rho$
and a noise state $\tau$ to result in a free state.
$R_\Free(\rho)$ is well defined for any $\rho$ as long as $\Free$ contains at least
one positive definite matrix, which we assume in what follows when referring to $R_\Free(\rho)$.
We discuss the case in which this assumption is not satisfied in Appendix~\ref{sec:EmptyInterior}.

%
\begin{thm} \label{thm:MainConvex}
 We have, for any convex subset $\Free$ of $\Den_A$,
 \begin{alignat}{1}
  G(\rho,\Free) &= 1 + R_\Free(\rho).
  \label{eq:main_convex}
 \end{alignat}
\end{thm}
\begin{proof}
 Let $\cl \Free$ denote the closure of $\Free$.
 Since $G(\rho,\Free) = G(\rho,\cl \Free)$ and $R_\Free(\rho) = R_{\cl \Free}(\rho)$ hold,
 we assume, without loss of generality, that $\Free$ is closed.

 First, we show $G(\rho,\Free) \le 1 + R_\Free(\rho)$.
 Since $\Free$ is closed, the lower bound in Eq.~\eqref{eq:R} can be replaced by a minimum.
 Therefore, letting $\lambda^\opt \coloneqq R_\Free(\rho)$, there exists
 $\tau^\opt \in \Den_A$ such that
 $(\rho + \lambda^\opt \tau^\opt) / (1 + \lambda^\opt) \eqqcolon \w^\opt \in \Free$.
 For any collection $\SetPL$ of channels $\Lambda_n$ with prior probabilities $p_n$
 and for any measurement $\Pi$, we have
 \begin{alignat}{1}
  \PS(\rho,\Pi;\SetPL) &=
  \sum_{n=1}^N p_n \Tr[\Pi_n \cdot \Lambda_n(\rho)] \nonumber \\
  &\le (1 + \lambda^\opt) \sum_{n=1}^N p_n \Tr[\Pi_n \cdot \Lambda_n(\w^\opt)] \nonumber \\
  &\le (1 + \lambda^\opt) \cdot \max_{\w \in \Free} \PS(\w;\SetPL),
  \label{eq:main_convex_le}
 \end{alignat}
 where the second line follows since the map
 $\sum_{n=1}^N p_n \Tr[\Pi_n \cdot \Lambda_n(\Endash)]$ is completely positive
 and $\rho \le (1 + \lambda^\opt) \w^\opt$ holds.
 Dividing the first and last equations by $\max_{\w \in \Free} \PS(\w;\SetPL)$,
 we obtain $G(\rho,\Free) \le 1 + \lambda^\opt = 1 + R_\Free(\rho)$.

 Next, we show $G(\rho,\Free) \ge 1 + R_\Free(\rho)$.
 It is known that $R_\Free(\rho)$ is the optimal value of the following optimization problem:
 \begin{alignat}{1}
  \begin{array}{ll}
   \mbox{maximize} & \Tr(x \rho) - 1 \\
   \mbox{subject~to} & x \ge \zero, \quad \Tr(x \w) \le 1 ~(\forall \w \in \Free),
  \end{array}
  \label{prob:R_dual}
 \end{alignat}
 which is obtained as the Lagrange dual problem of Eq.~\eqref{eq:R} (e.g., \cite{Reg-2017}).
 Let $x^\opt$ be the optimal solution of this problem, and let
 $e \coloneqq x^\opt / \|x^\opt\|_\infty$.
 Also, let $\w^\opt$ be $\w \in \Free$ that maximizes $\Tr(e \w)$.
 Choose an $N$-level system $B$,
 where $N$ satisfies the condition $N \ge 1 + 1/\Tr(e \w^\opt)$; then,
 we have, for each $\sigma \in \Den_A$,
 \begin{alignat}{1}
  \frac{1-\Tr(e \sigma)}{N-1} &\le \frac{1}{N-1} \le \Tr(e \w^\opt).
  \label{eq:N_we}
 \end{alignat}
 For each $n$, let $p_n \coloneqq 1/N$ and define the channel $\Lambda_n$ as
 \begin{alignat}{1}
  \Lambda_n(\sigma) &= \Tr(e \sigma) \ket{n}\bra{n}
  + \frac{1 - \Tr(e \sigma)}{N-1} (I_N - \ket{n}\bra{n})
  \label{eq:Lambda_n}
 \end{alignat}
 for any $\sigma \in \Den_A$,
 where $I_N$ is the identity matrix of order $N$
 and $\{ \ket{i} \}_{i=1}^N$ is the computational basis.
 Since $\Lambda_n(\sigma) = X^{n-1} \cdot \Lambda_1(\sigma) \cdot X^{1-n}$ holds
 for the generalized Pauli-$X$ matrix $X$,
 there exists a measurement $\Pi^\sym$ that satisfies
 $\Pi^\sym_n = X^{n-1} \Pi^\sym_1 X^{1-n}$ and
 $\PS(\sigma,\Pi;\SetPL) = \Tr[\Pi^\sym_1 \cdot \Lambda_1(\sigma)]$
 for any state $\sigma$ of $A$,
 proved in Appendix~\ref{sec:sym}.
 For such a measurement $\Pi^\sym$, let $c \coloneqq \braket{1|\Pi^\sym_1|1}$; then, we have
 \begin{alignat}{1}
  \PS(\sigma,\Pi;\SetPL) &= \Tr \left[ \Pi^\sym_1 \cdot \Lambda_1(\sigma) \right] \nonumber \\
  &= c \Tr(e \sigma) + (1 - c) \frac{1 - \Tr(e \sigma)}{N-1},
  \label{eq:PS_qc}
 \end{alignat}
 where we use $\Tr \Pi^\sym_1 = 1$ from $\Tr \Pi^\sym_n = \Tr \Pi^\sym_1$ and
 $\sum_{n=1}^N \Pi^\sym_n = I_N$.
 Therefore, we have, for each $\w \in \Free$,
 \begin{alignat}{1}
  \PS(\w,\Pi;\SetPL) &= c \Tr(e \w) + (1 - c) \frac{1 - \Tr(e \w)}{N-1} \le \Tr(e \w^\opt),
 \end{alignat}
 where the inequality follows from $\Tr(e \w) \le \Tr(e \w^\opt)$ and
 substituted $\sigma = \w$ into Eq.~\eqref{eq:N_we}.
 Therefore, we obtain $\PS(\w;\SetPL) \le \Tr(e \w^\opt)$.
 
 Also, for the measurement $\Pi' \coloneqq \{ \ket{n} \bra{n} \}_{n=1}^N$,
 we obtain $\PS(\rho,\Pi';\SetPL) = \Tr(e \rho)$ from Eq.~\eqref{eq:Lambda_n}.
 Thus, we have
 \begin{alignat}{1}
  \frac{\PS(\rho;\SetPL)}{\sup_{\w \in \Free} \PS(\w;\SetPL)}
  &\ge \frac{\Tr(e \rho)}{\Tr(e \w^\opt)}
  = \frac{\Tr(x^\opt \rho)}{\Tr(x^\opt \w^\opt)}
  \ge \Tr(x^\opt \rho),
 \end{alignat}
 where the last inequality follows from $\Tr(x^\opt \w^\opt) \le 1$.
 From this equation and $\Tr(x^\opt \rho) = 1 + R_\Free(\rho)$,
 we obtain $G(\rho,\Free) \ge 1 + R_\Free(\rho)$.
\end{proof}

It should be noted that when $\Free$ is convex, the robustness can be characterized by
a similar equation to Eq.~\eqref{eq:main_convex} \cite{Tak-Reg-Bu-Liu-2019}:
\begin{alignat}{1}
 \sup_{\SetPL} \sup_\Pi
 \frac{\PS(\rho,\Pi;\SetPL)}{\sup_{\w \in \Free} \PS(\w,\Pi;\SetPL)}
 &= 1 + R_\Free(\rho).
 \label{eq:H0}
\end{alignat}
This result has also been extended to infinite-dimensional systems,
nonconvex resource theories, and general probabilistic theories
\cite{Uol-Kra-Sha-Yu-2019,Reg-Lam-Fer-Tak-2021,Kur-Tak-Ade-Yam-2024,Tur-Gut-Ade-2024, Lam-Reg-Tak-Fer-2021}.
However, compared to the discrimination power, the measure represented by the left-hand side of
Eq.~\eqref{eq:H0} does not seem to be straightforwardly associated with an intuitive operational meaning
in the context of channel discrimination problems.
The differences between this measure and the discrimination power are discussed in detail in
Appendix~\ref{sec:GvsH}.

\section{The case of the set of free states being nonconvex}

Let $\co S$ denote the convex hull of a set $S$.
Then, the following corollary is immediately obtained.
\begin{cor} \label{cor:Main}
 We have, for any subset $\Free$ of $\Den_A$,
 \begin{alignat}{1}
  G(\rho,\Free) &= G(\rho,\co \Free) = 1 + R_{\co \Free}(\rho).
 \end{alignat}
\end{cor}
\begin{proof}
 The right-hand equality is clear from Theorem~\ref{thm:MainConvex}.
 To prove the left-hand equality, it is sufficient to show
 \begin{alignat}{1}
  \sup_{\w \in \Free} \PS(\w;\SetPL) &= \sup_{v \in \co \Free} \PS(v;\SetPL).
  \label{eq:main_vw}
 \end{alignat}
 From $\Free \subseteq \co \Free$, the left-hand side is clearly less than or equal to
 the right-hand side.
 The left-hand side is also greater than or equal to the right-hand side
 since any $v \in \co \Free$ can be expressed as a convex combination of
 some $\w_1, \w_2, \dots \in \Free$, and $\PS(v;\SetPL) \le \max_i \PS(\w_i;\SetPL)$ holds.
\end{proof}

In this paper, we denote a resource state $\rho \not\in \Free$
that satisfies $G(\rho,\Free) = 1$ as a bound resource state.
By this definition, a bound resource state can be seen as a resource state that is not useful
for channel discrimination.
From Corollary~\ref{cor:Main}, it is immediately clear that $\rho \not\in \Free$ being a bound resource
state is equivalent to $\rho \in \co \Free$.
In particular, there are no bound resource states if and only if $\Free$ is convex and closed.

\section{Extension to discrimination problems using auxiliary systems}

In the problem of distinguishing a collection of channels $\SetL_{n=1}^N$ from system $A$ to system $B$,
using an appropriate auxiliary system $C$ can improve the average success probability.
When the prior probabilities are $\SetP_{n=1}^N$, the average success probability
when a state $\rho$ is input into the composite system $A \ot C$ and
a measurement $\Pi \coloneqq \SetPi_{n=1}^N$ is performed on the composite system $B \ot C$
is given by
\begin{alignat}{1}
 \PS(\rho,\Pi;\SetPL) &\coloneqq \sum_{n=1}^N p_n \Tr[\Pi_n \cdot (\Lambda_n \ot \id_C)(\rho)]
 \label{eq:PS_rho_Pi_ex}
\end{alignat}
instead of Eq.~\eqref{eq:PS_rho_Pi}.

Given a state $\rho \in \Den_{A \ot C}$ and the set of free states $\Free$
of the composite system $A \ot C$,
let $G_C(\rho,\Free)$ denote the discrimination power $G(\rho,\Free)$
defined by Eq.~\eqref{eq:G} using the average success probability in Eq.~\eqref{eq:PS_rho_Pi_ex}.
Let us consider determining this discrimination power.
Note that when $C$ is a quantum $1$-level system (i.e., no auxiliary system),
$G_C(\rho,\Free)$ and $G(\rho,\Free)$ are clearly equal.
For a set $S$ of some states of the composite system $A \ot C$,
if for any $\w \in S$ and any channel $\cE$ on $C$,
the state $(\id_A \ot \cE)(\w)$ is in $S$,
then we say that $S$ preserves operations on $C$.
If $\Free$ preserves operations on $C$, then it means that applying any local channel
to system $C$ does not turn a free state into a resource state.
\begin{thm} \label{thm:ExLocal}
 If a convex subset $\Free$ of $\Den_{A \ot C}$ preserves operations on $C$,
 then we have
 \begin{alignat}{1}
  G_C(\rho,\Free) &= 1 + R_\Free(\rho).
  \label{eq:main_ext_local}
 \end{alignat}
\end{thm}
\begin{proof}
 $G_C(\rho,\Free)$ can be considered the discrimination power when a collection of channels
 from $A \ot C$ to $B \ot C$ is limited to those of the form $\{ \Lambda_n \ot \id_C \}_{n=1}^N$.
 Therefore, $G_C(\rho,\Free)$ cannot be greater than $G(\rho,\Free)$,
 and thus $G_C(\rho,\Free) \le G(\rho,\Free) = 1 + R_\Free(\rho)$ holds
 from Theorem \ref{thm:MainConvex}.
 It remains to show that $G_C(\rho,\Free) \ge 1 + R_\Free(\rho)$.
 Here, we present an outline; the full proof is given in Appendix~\ref{sec:ExLocal}.

 Let $x^\opt$ be the optimal solution of Problem~\eqref{prob:R_dual}.
 For a non-normalized maximally entangled state
 $\ket{\Phi} \coloneqq \sum_{n=1}^{N_C} \ket{n} \ot \ket{n}$,
 where $N_C$ is the level of $C$, we can choose a positive real number $c$
 and subchannels $\{ \tLambda_n \}_{n=1}^N$
 (i.e., completely positive maps such that $\sum_n \tLambda_n$ is a channel)
 that satisfy
 \begin{alignat}{1}
  \braket{\Phi|(\tLambda_1 \ot \id_C)(\rho)|\Phi} = c \Tr(x^\opt \rho),
  \quad \forall \rho \in \Den_{A \ot C}.
  \label{eq:ExLocalPhiTr}
 \end{alignat}
 We consider the collection of channels $\{ \Lambda_{i,q,r} \}_{(i,q,r)=(1,1,1)}^{(N,N_C,N_C)}$
 defined by
 \begin{alignat}{1}
  \Lambda_{i,q,r}(\Endash) &\coloneqq \sum_{n=0}^{N-1} \ket{n+i}\bra{n+i} \ot X^{q-1}Z^{r-1}
  \tLambda_{n+1}(\Endash) Z^{1-r}X^{1-q},
 \end{alignat}
 where $\ket{N+1} \coloneqq \ket{1}, \dots, ~\ket{2N-1} \coloneqq \ket{N-1}$,
 and $X$ (resp. $Z$) is the generalized Pauli-$X$ (resp. Pauli-$Z$) matrix.
 Then, it can be seen that
 \begin{alignat}{1}
  N_C \PS(\rho;\SetPLIQR) &\ge \braket{\Phi|(\tLambda_1 \ot \id_C)(\rho)|\Phi}
 \end{alignat}
 holds for equal prior probabilities
 $\{ p_{i,q,r} \coloneqq 1 / (N N_C^2) \}_{(i,q,r)=(1,1,1)}^{(N,N_C,N_C)}$.
 It can also be shown that for any positive real number $\varepsilon$ and any free state $\w$,
 there exists a channel $\Gamma$ satisfying
 \begin{alignat}{1}
  N_C \PS(\w;\SetPLIQR) &\le \braket{\Phi|(\tLambda_1 \ot \Gamma)(\w)|\Phi}
  + c \varepsilon \nonumber \\
  &\le \sup_{\w' \in \Free} \braket{\Phi|(\tLambda_1 \ot \id_C)(\w')|\Phi}
  + c \varepsilon,
 \end{alignat}
 where the second line follows from $(\id_A \ot \Gamma)(\w)$ being a free state.
 Therefore, we have, from Eq.~\eqref{eq:ExLocalPhiTr},
 \begin{alignat}{1}
  G_C(\rho,\Free)
  &\ge \frac{c \Tr(x^\opt \rho)}{c \sup_{\w' \in \Free }\Tr(x^\opt \w') + c \varepsilon}
  \ge \frac{1 + R_\Free(\rho)}{1 + \varepsilon}.
 \end{alignat}
 Considering the limit as $\varepsilon \to 0$, we obtain $G_C(\rho,\Free) \ge 1 + R_\Free(\rho)$.
\end{proof}

We provide some examples of sets that preserve operations on $C$.
An example is a set that is invariant under local operations and classical communication (LOCC).
One such example is the set of all $\w \in \Den_{A \ot C}$
satisfying $\mu(\w) \le c$, where $c$ is a constant and $\mu$ is an entanglement measure
that does not increase under LOCC, such as the Schmidt number \cite{Ter-Hor-2000},
the entanglement of formation \cite{Ben-Div-Smo-Woo-1996}, the logarithmic negativity
\cite{Vid-Wer-2002,Ple-2005},
and the concurrence \cite{Hil-Woo-1997}.
The set of all separable states can be considered a special case of this.
Other examples include the set of unsteerable states
(from $A$ to $C$ or $C$ to $A$) \cite{Wis-Jon-Doh-2007}
and the set of Bell local states \cite{Bell-1964}.
In Ref.~\cite{Bae-Chr-Pia-2019}, it is shown that for the set, $S_k$, of states with Schmidt number
at most $k$, $[1 + R_{S_k}(\rho)] / k$ is a lower bound for $G_C(\rho,S_k)$.
In contrast, Theorem~\ref{thm:ExLocal} shows that $G_C(\rho,S_k)$ is exactly $1 + R_{S_k}(\rho)$.
Additionally, Ref.~\cite{Pia-Wat-2015} defines a different robustness measure for steerability
(where $\Free$ is the set of unsteerable states) and shows that it can quantify
the discrimination power in subchannel discrimination problems.
Theorem~\ref{thm:ExLocal} implies a more straightforward quantification.

\begin{figure}[t]
 \centering
 \includegraphics[scale=1.0]{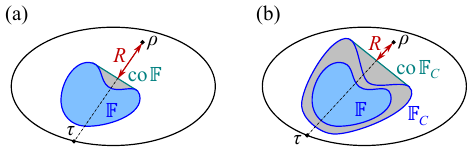}
 \caption{The robustness of $\rho \in \Den_{A \ot C}$ with respect to
 (a) $\co \Free$ and (b) $\co \Free_C$.
 $R_{\co \Free}(\rho) \ge R_{\co \Free_C}(\rho)$ holds from $\Free \subseteq \Free_C$.
 From Corollaries~\ref{cor:Main} and \ref{cor:ExLocalGeneral}, these values respectively quantify
 the discrimination power $G(\rho,\Free)$ and $G_C(\rho,\Free)$.
 The gray areas, $\co \Free \setminus \Free$ and $\co \Free_C \setminus \Free$,
 represent the sets of bound resource states in the sense of $G$ and $G_C$,
 respectively.
 Note that $R_{\co \Free}(\rho) = R_{\co \Free_C}(\rho)$, i.e.,
 $G(\rho,\Free) = G_C(\rho,\Free)$, holds for each $\rho \in \Den_{A \ot C}$
 if and only if $\co \Free = \co \Free_C$ holds.}
 \label{fig:robustness}
\end{figure}

Given a subset $\Free$ of $\Den_{A \ot C}$, let $\Free_C$ denote
the set of all states of the form $(\id_A \ot \cE)(\w)$, where $\w$ is in $\Free$ and $\cE$
is a channel on $C$.
Then, the following corollary is obtained from Theorem~\ref{thm:ExLocal}.
\begin{cor} \label{cor:ExLocalGeneral}
 We have, for any subset $\Free$ of $\Den_{A \ot C}$,
 \begin{alignat}{1}
  G_C(\rho,\Free) &= G_C(\rho,\Free_C) = 1 + R_{\co \Free_C}(\rho).
  \label{eq:main_ext_local_general}
 \end{alignat}
\end{cor}
\begin{proof}
 Using the same method as in Corollary~\ref{cor:Main},
 we find $G_C(\rho,\Free_C) = G_C(\rho,\co \Free_C)$.
 Each $v \in \co \Free_C$ can be expressed in the form $\sum_i q_i (\id_A \ot \cE_i)(\w_i)$,
 where $\w_1, \w_2, \dots$ are in $\Free$, $\cE_1, \cE_2, \dots$ are channels on $C$,
 and $\{ q_i \}$ is a probability distribution.
 Therefore, for any channel $\cE'$ on $C$,
 $(\id_A \ot \cE')(v) = \sum_i q_i (\id_A \ot \cE' \circ \cE_i)(\w_i) \in \co \Free_C$,
 and thus $\co \Free_C$ preserves operations on $C$.
 From Theorem~\ref{thm:ExLocal}, the right-hand equality
 in Eq.~\eqref{eq:main_ext_local_general} holds.
 To show the left-hand equality in Eq.~\eqref{eq:main_ext_local_general},
 it is sufficient to show that
 \begin{alignat}{1}
  \sup_{\w \in \Free} \PS(\w;\SetPL) &= \sup_{v \in \Free_C} \PS(v;\SetPL).
  \label{eq:main_vw2}
 \end{alignat}
 From $\Free \subseteq \Free_C$,
 the left-hand side is clearly less than or equal to the right-hand side.
 The left-hand side is also greater than or equal to the right-hand side
 since for each $v \in \Free_C$ expressed as $(\id_A \ot \cE)(\w)$
 (where $\w$ is in $\Free$ and $\cE$ is a channel on $C$),
 we have, for any measurement $\Pi$,
 \begin{alignat}{1}
  \PS(v,\Pi;\SetPL)
  &= \sum_{n=1}^N p_n \Tr[\Pi_n \cdot (\Lambda_n \ot \id_C)[(\id_A \ot \cE)(\w)]] \nonumber \\
  &= \sum_{n=1}^N p_n \Tr[(\id_B \ot \cE)^\dagger(\Pi_n) \cdot (\Lambda_n \ot \id_C)(\w)] \nonumber \\
  &\le \PS(\w;\SetPL),
 \end{alignat}
 where $\dagger$ denotes the adjoint.
 The inequality follows from $\{ (\id_B \ot \cE)^\dagger(\Pi_n) \}_{n=1}^N$ being a measurement.
\end{proof}

Equation~\eqref{eq:main_ext_local_general} allows the discrimination power
using auxiliary systems to be formulated in terms of the robustness.
This corollary also shows that $\rho \not\in \Free$ is a bound resource state
in the sense of $G_C$, i.e., $G_C(\rho,\Free) = 1$,
if and only if $\rho$ is in $\co \Free_C$.
Furthermore, Corollaries~\ref{cor:Main} and \ref{cor:ExLocalGeneral} give
$G_C(\rho,\Free) = 1 + R_{\co \Free_C}(\rho) = G(\rho,\Free_C)$.
Figure~\ref{fig:robustness} schematically represents $R_{\co \Free}(\rho)$ and $R_{\co \Free_C}(\rho)$.

\section{Conclusion}

In this study, we have quantified the discrimination power in quantum channel discrimination problems
for any resource theory of states.
First, we have quantified the discrimination power of any resource state using the robustness measure
in channel discrimination problems without auxiliary systems.
This means that the robustness can be fully characterized as a value that can be interpreted
operationally as discrimination power of channels.
Next, we have shown that in channel discrimination problems using an auxiliary system $C$,
the discrimination power can similarly be expressed using the robustness measure
with respect to the set $\co \Free_C$.
This result is expected to be useful in analyzing the discrimination power
in channel discrimination problems involving auxiliary systems.

In our study, we have considered a single-shot channel discrimination strategy
that maximizes the average success probability.
Future developments of this research could involve considering other optimization criteria,
such as an unambiguous criterion or a minimax criterion, or multi-shot strategies
to clarify the measures representing discrimination power.
Clarifying such measures would deepen our understanding of the relationship between
quantum channel discrimination problems and quantum resource theories,
promoting further theoretical advancements.

\section*{Acknowledgment}

We thank for O.~Hirota, M.~Sohma, T.~S.~Usuda, and K.~Kato for insightful discussions.
This work was supported by the Air Force Office of Scientific Research under
award number FA2386-22-1-4056.


\appendix

\section{Preliminaries}

For each quantum system $A$, let $N_A$ be the level of $A$.
The set of all positive semidefinite matrices of system $A$ is denoted by $\Pos_A$,
where square matrices of order $N_A$ are referred to as square matrices of $A$.
When two Hermitian matrices $X$ and $Y$ of $A$ satisfy $X - Y \in \Pos_A$,
we write $X \ge Y$ or $Y \le X$.
In particular, $X \ge \zero$ (where $\zero$ is the zero matrix) is equivalent to $X \in \Pos_A$.
The sets of all completely positive (CP) maps, all channels, i.e., trace-preserving CP maps,
and all trace-nonicreasing CP maps (called TNI-CP maps) from system $A$ to system $B$
are denoted by $\CP_{A \to B}$, $\Chn_{A \to B}$, and $\TNICP_{A \to B}$, respectively.
A map from $A$ to $A$ is called a map on $A$,
and the identity channel on $A$ is denoted by $\id_A$.
CP maps $\{ \Lambda_n \in \CP_{A \to B} \}_{n=1}^N$
(where $N$ is a natural number) that satisfy $\sum_{n=1}^N \Lambda_n \in \Chn_{A \to B}$
are called subchannels.

A map $f$ is often written as $f(\Endash)$.
For example, $\Tr(e \cdot \Endash)$ with $e \in \Pos_A$ refers to
the map $\Pos_A \ni x \mapsto \Tr(e x) \in \Realp$,
where $\Realp$ is the set of all non-negative real numbers,
and in particular, $\Tr(I_{N_A} \cdot \Endash)$ can be written as $\Tr \Endash$.
A CP map $\Lambda$ is a channel if and only if it satisfies
$\Tr[\Lambda(\Endash)] = \Tr \Endash$.

The computational basis of each system $A$ is denoted by $\{ \ket{n} \}_{n=1}^{N_A}$.
The unitary matrix $X$ of order $N_A$ defined by
$X \ket{n} = \ket{n + 1}$ $(\forall n \in \{1,\dots,N_A-1\})$ and $X \ket{N_A} = \ket{1}$
is called the \termdef{generalized Pauli-$X$ matrix}.
The unitary matrix $Z$ of order $N_A$ defined by
$Z \ket{n} = \exp[2 \pi n \sqrt{-1}/N_A] \ket{n}$ $~(\forall n \in \{1,\dots,N_A\})$
is called the \termdef{generalized Pauli-$Z$ matrix}.

Let $\Phi_A$ be the non-normalized maximally entangled state of the composite system $A \ot A$
denoted by
\begin{alignat}{2}
 \Phi_A &\coloneqq \ket{\Phi}\bra{\Phi} \in \Pos_{A \ot A},
 &\quad \ket{\Phi} &\coloneqq \sum_{n=1}^{N_A} \ket{n} \ot \ket{n}.
\end{alignat}
We have, for each $\rho, \sigma \in \Pos_A$,
\begin{alignat}{1}
 \Tr \left[ \Phi_A \cdot (\rho \ot \sigma) \right] &= \Tr(\sigma^\T \rho),
 \label{eq:Phi_T}
\end{alignat}
where $^\T$ is the transpose.
Also, we have, for each $\rho, \sigma \in \Pos_{A \ot C}$,
\begin{alignat}{1}
 \Tr \left[ \Phi_C \cdot \left[ [\Tr(\sigma \cdot \Endash) \ot \id_C] \c (\id_A \ot \Phi_C)
 \ot \id_C \right] (\rho) \right] &= \Tr(\sigma \rho)
 \label{eq:Phi_zigzag}
\end{alignat}
since, for each $\rho_1 \in \Pos_A$ and $\rho_2 \in \Pos_C$,
\begin{alignat}{1}
 &\hspace{-1em} \Tr \left[ \Phi_C \cdot [[\Tr(\sigma \cdot \Endash) \ot \id_C] \c (\id_A \ot \Phi_C) \ot \id_C]
 (\rho_1 \ot \rho_2) \right] \nonumber \\
 &= \Tr \left[ \rho_2^\T \cdot [\Tr[\sigma \cdot (\rho_1 \ot \Endash)] \ot \id_C](\Phi_C) \right] \nonumber \\
 &= \Tr \left[ (\sigma' \ot \rho_2^\T) \cdot \Phi_C \right] \nonumber \\
 &= \Tr(\sigma' \rho_2) \nonumber \\
 &= \Tr[\sigma (\rho_1 \ot \rho_2)]
\end{alignat}
holds, where $\sigma' \coloneqq \Tr_A[\sigma (\rho_1 \ot I_{N_C})]$
and the first and third equalities follow from Eq.~\eqref{eq:Phi_T}.

\section{Group covariant measurements} \label{sec:sym}

For any system $B$ and any finite group $G$, let us consider a collection of
unitary matrices $\{ U_g \}_{g \in G}$ of order $N_B$ satisfying $U_e = I_{N_B}$
for the identity element $e$ of $G$, and $\cU_g \circ \cU_h = \cU_{gh}$ for all $g,h \in G$,
where $\cU_g$ is $U_g (\Endash) U_g^{-1} \in \Chn_{B \to B}$, i.e., the map
$\Den_B \ni \rho \mapsto U_g \rho U_g^{-1} \in \Den_B$.
Let $\cG \coloneqq \{ \cU_g \}_{g \in G}$;
then, $\cG$ can be seen as an action of $G$ on $\Den_B$.
$\cG$ itself is also a group with composition $\circ$ as the operation
and $\cU_e = \id_B$ as its identity element.

Let $|G|$ denote the number of elements in $G$, and consider a collection of channels
$\{ \Lambda_g \}_{g \in G}$ from $A$ to $B$
that satisfies $\Lambda_{gh} = \cU_g \circ \Lambda_h$ for all $g,h \in G$.
Also, let us consider the equal prior probabilities $\{ p_g \coloneqq |G|^{-1} \}_{g \in G}$;
then, the collection $\{ p_g, \Lambda_g \}_{g \in G}$
of channels $\Lambda_g$ with prior probabilities $p_g$ (or simply the collection $\{ \Lambda_g \}$)
is called \termdef{$\cG$-covariant}.
Additionally, a measurement $\Pi \coloneqq \{ \Pi_g \}_{g \in G}$ on system $B$
is called \termdef{$\cG$-covariant} if it satisfies $\Pi_{gh} = \cU_g(\Pi_h)$ for all $g,h \in G$.

\begin{ex} \label{ex:cG_cycle}
 In the proof of Theorem~\ref{thm:MainConvex}, we consider a collection of channels
 $\{ \Lambda_n \}_{n=1}^N$ that satisfies
 $\Lambda_n(\sigma) = X^{n-1} \cdot \Lambda_1(\sigma) \cdot X^{1-n}$.
 For an integer $n$, let
 \begin{alignat}{1}
  \mod{n}{N} &\coloneqq
  \begin{dcases}
   n, & 1 \le n \le N, \\
   \mod{n - N}{N}, & n > N, \\
   \mod{n + N}{N}, & n < 1, \\
  \end{dcases}
  \label{eq:mod}
 \end{alignat}
 i.e., $\mod{n}{N}$ is $n' \in \{1,2,\dots,N\}$ such that $n - n'$ is divisible by $N$.
 Consider the cyclic group $G \coloneqq \{ 1,2,\dots,N \}$,
 where the product of $m,n \in G$ is $\mod{m + n}{N}$, and the identity element of $G$ is $N$.
 Let $\cG \coloneqq \{ \cU_k \coloneqq X^k (\Endash) X^{-k} \}_{k \in G}$; then,
 $\{ \Lambda_n \}_{n=1}^N$ is $\cG$-covariant
 since $\Lambda_{\mod{m+n}{N}} = \cU_m \circ \Lambda_n$ holds for all $m,n \in G$.
 Note that $\cU_N = \id_B$ holds from $X^N = I_N$.
\end{ex}

The following lemma holds.
\begin{lemma} \label{lemma:Sym}
 For any two systems $A$ and $B$, consider a $\cG$-covariant collection
 $\{ p_g, \Lambda_g \}_{g \in G}$ of channels $\Lambda_g$ from $A$ to $B$
 with prior probabilities $p_g$.
 For any measurement $\Pi \coloneqq \{ \Pi_g \}_{g \in G}$ on system $B$,
 let us define the measurement $\Pi^\sym \coloneqq \{ \Pi^\sym_g \}_{g \in G}$ as
 \begin{alignat}{1}
  \Pi^\sym_g &\coloneqq \frac{1}{|G|} \sum_{h \in G} \cU_h(\Pi_{h^{-1}g}),
  \label{eq:Pisym}
 \end{alignat}
 where $h^{-1}$ is the inverse of $h \in G$.
 Then, $\Pi^\sym$ is $\cG$-covariant, and we have, for any state $\sigma$ of $A$,
 \begin{alignat}{1}
  \PS(\sigma,\Pi;\{ p_g,\Lambda_g\}) &= \PS(\sigma,\Pi^\sym;\{ p_g,\Lambda_g\})
  \nonumber \\
  &= \Tr \left[ \Pi^\sym_{g'} \cdot \Lambda_{g'}(\sigma) \right], \quad \forall g' \in G.
  \label{eq:PS_Pi_Pisym}
 \end{alignat}
\end{lemma}
\begin{proof}
 It is clear that $\Pi^\sym$ is a measurement since
 \begin{alignat}{1}
  \sum_{g \in G} \Pi^\sym_g
  &= \frac{1}{|G|} \sum_{g \in G} \sum_{h \in G} \cU_h(\Pi_{h^{-1}g})
  = \frac{1}{|G|} \sum_{h \in G} \cU_h \left( \sum_{g \in G} \Pi_{h^{-1}g} \right)
  \nonumber \\
  &= \frac{1}{|G|} \sum_{h \in G} \cU_h(I_{N_B})
  = I_{N_B}
 \end{alignat}
 and $\Pi^\sym_g \ge \zero$ $~(\forall g \in G)$ hold.
 Also, we have, for each $g,h \in G$,
 \begin{alignat}{1}
  \Pi^\sym_{gh} &= \frac{1}{|G|} \sum_{k \in G} \cU_k(\Pi_{k^{-1}gh})
  = \frac{1}{|G|} \sum_{k' \in G} \cU_{gk'}(\Pi_{k'^{-1}h})
  \nonumber \\
  &= \cU_g \left[ \frac{1}{|G|} \sum_{k' \in G} \cU_{k'}(\Pi_{k'^{-1}h}) \right]
  = \cU_g(\Pi^\sym_h),
 \end{alignat}
 where $k' \coloneqq g^{-1}k$.
 Thus, $\Pi^\sym$ is $\cG$-covariant.
 We have, for any $\sigma \in \Den_A$ and $g' \in G$,
 \begin{alignat}{1}
  \Tr \left[ \Pi^\sym_{g'} \cdot \Lambda_{g'}(\sigma) \right]
  &= \frac{1}{|G|} \sum_{h \in G} \Tr[\cU_h(\Pi_{h^{-1}g'}) \cdot \Lambda_{g'}(\sigma)] \nonumber \\
  &= \frac{1}{|G|} \sum_{h \in G} \Tr[\Pi_{h^{-1}g'} \cdot \cU_{h^{-1}}[\Lambda_{g'}(\sigma)]] \nonumber \\
  &= \frac{1}{|G|} \sum_{h \in G} \Tr[\Pi_{h^{-1}g'} \cdot \Lambda_{h^{-1}g'}(\sigma)] \nonumber \\
  &= \PS(\sigma,\Pi;\{ p_g,\Lambda_g\}),
  \label{eq:Pisym_proof1}
 \end{alignat}
 where the second equality follows from $\Tr[\cU_h(x) \cdot y] = \Tr[x \cdot \cU_{h^{-1}}(y)]$
 for any $x,y \in \Pos_B$.
 Therefore, we obtain
 \begin{alignat}{1}
  \PS(\sigma,\Pi^\sym;\{ p_g,\Lambda_g\})
  &= \frac{1}{|G|} \sum_{g' \in G} \Tr \left[ \Pi^\sym_{g'} \cdot \Lambda_{g'}(\sigma) \right]
  \nonumber \\
  &= \frac{1}{|G|} \sum_{g' \in G} \PS(\sigma,\Pi;\{ p_g,\Lambda_g\}) \nonumber \\
  &= \PS(\sigma,\Pi;\{ p_g,\Lambda_g\}).
  \label{eq:Pisym_proof2}
 \end{alignat}
 Equations~\eqref{eq:Pisym_proof1} and \eqref{eq:Pisym_proof2} yield Eq.~\eqref{eq:PS_Pi_Pisym}.
\end{proof}

\section{Proof of Theorem~\ref{thm:ExLocal}} \label{sec:ExLocal}

After giving two lemmas, we will prove Theorem~\ref{thm:ExLocal}.
Let us consider two systems $A$ and $C$ and a subset $S$ of $\Den_{A \ot C}$.
If for any $\w \in S$ and any TNI-CP map $\cE$ on $C$,
the non-normalized state $\upsilon \coloneqq (\id_A \ot \cE)\w$, once normalized,
is in $S$ or $\upsilon = \zero$,
then we say that $S$ \termdef{preserves probabilistic transformations} on $C$.

\begin{lemma} \label{lemma:ExLocalGain}
 Arbitrarily choose two systems $A$ and $C$ and a positive real number $\varepsilon$.
 If the set of free states $\Free$ of the composite system $A \ot C$ is convex and
 preserves operations on $C$, then there exist a natural number $N$ and
 a collection of subchannels $\{ \tLambda_n \}_{n=1}^N$ from $A$ to $C$ such that
 \begin{alignat}{1}
  \frac{\Tr \left[ \Phi_C \cdot (\tLambda_1 \ot \id_C)(\rho) \right]}{
  \sup_{\w \in \Free} \sum_{n=1}^N \Tr \left[ \Phi_C \cdot (\tLambda_n \ot \tGamma_n)(\w) \right]}
  &\ge \frac{1 + R_\Free(\rho)}{1 + \varepsilon}
  \label{eq:ExLocalGain}
 \end{alignat}
 holds for any subchannels $\{ \tGamma_n \}_{n=1}^N$ on $C$.
 In particular, if $\Free$ preserves probabilistic transformations on $C$, then
 the same holds for $\varepsilon = 0$.
\end{lemma}
\begin{proof}
 First, let us consider the case in which $\Free$ preserves operations on $C$.
 Let $x^\opt$ be the optimal solution of Problem~\ref{prob:R_dual}.
 Let $\w^\opt$ be $\w \in \cl \Free$ that maximizes $\Tr(x^\opt \w)$.
 Also, let
 \begin{alignat}{1}
  e &\coloneqq \frac{x^\opt}{\|\Tr_C x^\opt\|_\infty}, \nonumber \\
  \te &\coloneqq [\Tr(e \cdot \Endash) \ot \id_C] \c (\id_A \ot \Phi_C), \nonumber \\
  a &\coloneqq (\Tr \Endash) - \Tr[\te(\Endash)]
 \end{alignat}
 and let $\tau$ be any state of $C$.
 Note that since $\te$ is a TNI-CP map, so is $a$.
 Since $\Tr(x^\opt \w) \le \Tr(x^\opt \w^\opt)$ holds for any $\w \in \Free$,
 $\Tr(e \w) \le \Tr(e \w^\opt)$ holds for any $\w \in \Free$.
 Let $\varepsilon' \coloneqq \varepsilon / \|\Tr_C x^\opt\|_\infty$,
 and choose a natural number $N$ such that
 \begin{alignat}{1}
  \frac{1}{N-1} \Tr[\Phi_C \cdot (\tau \c a \ot \id_C)(\sigma)] &\le \Tr(e \sigma) + \varepsilon',
  \nonumber \\
  &\quad \forall \sigma \in \Den_{A \ot C}
  \label{eq:ExLocalN}
 \end{alignat}
 holds.
 Since the right-hand side is positive, such $N$ exists.
 Note that since $\Tr[\Phi_C \cdot (\tau \circ a \ot \id_C)(\sigma)] \le 1$
 and $\Tr(e \sigma) \ge 0$ hold, a stricter condition $\frac{1}{N-1} \le \varepsilon'$ can be
 used to choose $N$.
 Define $\{ \tLambda_n \}_{n=1}^N$ as
 \begin{alignat}{1}
  \tLambda_n &\coloneqq
  \begin{dcases}
   \te, & n = 1, \\
   \frac{1}{N-1} (\tau \circ a), & n \ge 2; \\
  \end{dcases}
 \end{alignat}
 then, since $\sum_{n=1}^N \Tr[\tLambda_n(\Endash)] = \Tr[\te(\Endash)] + a = \Tr \Endash$ holds,
 $\sum_{n=1}^N \tLambda_n$ is a channel, and thus $\{ \tLambda_n \}_{n=1}^N$ are subchannels.

 We have, for any $\sigma \in \Den_{A \ot C}$ and any CP map $\cE$ on $C$,
 \begin{alignat}{1}
  &\hspace{-1em} \Tr \left[ \Phi_C \cdot (\tLambda_1 \ot \cE)(\sigma) \right] \nonumber \\
  &= \Tr \left[ \Phi_C \cdot (\te \ot \cE)(\sigma) \right] \nonumber \\
  &= \Tr \left[ \Phi_C \cdot \left[ [\Tr(e \cdot \Endash) \ot \id_C] \c (\id_A \ot \Phi_C)
  \ot \id_C \right] \right. \nonumber \\
  &\quad \left. [(\id_A \ot \cE)(\sigma)] \right] \nonumber \\
  &= \Tr[e \cdot (\id_A \ot \cE)(\sigma)],
  \label{eq:ExLocalGain_e_sigma}
 \end{alignat}
 where the last equality follows from Eq.~\eqref{eq:Phi_zigzag}.
 Substituting $\sigma = \rho$ and $\cE = \id_C$ into this equation,
 we find that the numerator on the left-hand side of Eq.~\eqref{eq:ExLocalGain} satisfies
 \begin{alignat}{1}
  \Tr \left[ \Phi_C \cdot (\tLambda_1 \ot \id_C)(\rho) \right]
  &= \Tr(e \rho).
  \label{eq:ExLocalGain1}
 \end{alignat}
 Also, the denominator of the left-hand side of Eq.~\eqref{eq:ExLocalGain}
 is at most $\Tr(e \w^\opt) + \varepsilon'$.
 Indeed, we have, for any $\w \in \Free$,
 \begin{alignat}{1}
  &\hspace{-1em} \sum_{n=1}^N \Tr \left[ \Phi_C \cdot (\tLambda_n \ot \tGamma_n)(\w) \right] \nonumber \\
  &= \Tr \left[ \Phi_C \cdot (\tLambda_1 \ot \tGamma_1)(\w) \right] \nonumber \\
  &\quad + \frac{1}{N-1} \Tr \left[ \Phi_C \cdot (\tau \c a \ot \tGamma_0)(\w) \right]
  \nonumber \\
  &= \Tr \left[ e \cdot (\id_A \ot \tGamma_1)(\w) \right] \nonumber \\
  &\quad + \frac{1}{N-1} \Tr \left[ \Phi_C \cdot (\tau \c a \ot \id_C)
  \left[ (\id_A \ot \tGamma_0)(\w) \right] \right] \nonumber \\
  &\le \Tr \left[ e \cdot (\id_A \ot \tGamma_1)(\w) \right]
  + \Tr \left[ e \cdot (\id_A \ot \tGamma_0)(\w) \right] + \varepsilon' \nonumber \\
  &= \Tr \left[ e \cdot \left[ \id_A \ot (\tGamma_1 + \tGamma_0) \right](\w) \right] + \varepsilon'
  \nonumber \\
  &\le \Tr(e \w^\opt) + \varepsilon',
  \label{eq:ExLocalGain2}
 \end{alignat}
 where $\tGamma_0 \coloneqq \sum_{n=2}^N \tGamma_n$.
 The second equality follows from Eq.~\eqref{eq:ExLocalGain_e_sigma}
 with $\sigma$ substituted by $\w$ and $\cE$ substituted by $\tGamma_1$.
 The first inequality follows from Eq.~\eqref{eq:ExLocalN}
 with $\sigma$ substituted by $(\id_A \ot \tGamma_0)(\w)$, and
 the second inequality follows from $\left[ \id_A \ot (\tGamma_1 + \tGamma_0) \right](\w) \in \Free$,
 derived from the fact that $\Free$ preserves operations on $C$
 and $\tGamma_1 + \tGamma_0$ is a channel.
 Therefore, from Eqs.~\eqref{eq:ExLocalGain1} and \eqref{eq:ExLocalGain2}, we obtain
 \begin{alignat}{1}
  \frac{\Tr \left[ \Phi_C \cdot (\tLambda_1 \ot \id_C)(\rho) \right]}{
  \sup_{\w \in \Free} \sum_{n=1}^N \Tr \left[ \Phi_C \cdot (\tLambda_n \ot \tGamma_n)(\w) \right]}
  &\ge \frac{\Tr(e \rho)}{\Tr(e \w^\opt) + \varepsilon'} \nonumber \\
  &= \frac{\Tr(x^\opt \rho)}{\Tr(x^\opt \w^\opt) + \varepsilon} \nonumber \\
  &\ge \frac{1 + R_\Free(\rho)}{1 + \varepsilon},
  \label{eq:ExLocalGainResult}
 \end{alignat}
 where the last inequality follows from $\Tr(x^\opt \rho) = 1 + R_\Free(\rho)$
 and $\Tr(x^\opt \w^\opt) \le 1$.

 Next, let us consider the case in which $\Free$ preserves probabilistic transformations on $C$.
 This is similar to the case in which $\Free$ preserves operations on $C$,
 but instead of Eq.~\eqref{eq:ExLocalN}, we choose $N$ such that
 \begin{alignat}{1}
  \frac{1}{N-1} \Tr[\Phi_C \cdot (\tau \c a \ot \id_C)(\w)] \le \Tr(e \w^\opt),
  \quad \forall \w \in \Free \nonumber \\
  \label{eq:ExLocalNCP}
 \end{alignat}
 holds.
 Note that since $\Tr[\Phi_C \cdot (\tau \circ a \ot \id_C)(\w)] \le 1$ holds,
 a stricter condition $\frac{1}{N-1} \le \Tr(e \w^\opt)$ can be used to choose $N$.
 We have
 \begin{alignat}{1}
  \hspace{1em} &\hspace{-1em} \sum_{n=1}^N \Tr \left[ \Phi_C \cdot (\tLambda_n \ot \tGamma_n)(\w) \right] \nonumber \\
  &= \Tr \left[ \Phi_C \cdot (\tLambda_1 \ot \tGamma_1)(\w) \right]
  + \frac{1}{N-1} \Tr \left[ \Phi_C \cdot (\tau \c a \ot \tGamma_0)(\w) \right]
  \nonumber \\
  &= \Tr \left[ e \cdot (\id_A \ot \tGamma_1)(\w) \right] \nonumber \\
  &\quad + \frac{1}{N-1} \Tr \left[ \Phi_C \cdot (\tau \c a \ot \id_C)
  \left[ (\id_A \ot \tGamma_0)(\w) \right] \right] \nonumber \\
  &= \alpha \Tr(e \w')
  + \frac{1 - \alpha}{N-1} \Tr \left[ \Phi_C \cdot (\tau \c a \ot \id_C)(\w'') \right]
  \nonumber \\
  &\le \alpha \Tr(e \w^\opt) + (1 - \alpha) \Tr(e \w^\opt) \nonumber \\
  &= \Tr(e \w^\opt),
  \label{eq:ExLocalGain2CP}
 \end{alignat}
 where the first two equalities are the same as Eq.~\eqref{eq:ExLocalGain2},
 and let $\alpha \coloneqq \Tr[(\id_A \ot \tGamma_1)(\w)]$.
 In the third equality, we chose $\w' \in \Free$ such that $(\id_A \ot \tGamma_1)(\w) = \alpha \w'$ holds,
 and similarly, $\w'' \in \Free$ such that $(\id_A \ot \tGamma_0)(\w) = (1 - \alpha) \w''$ holds.
 Since $\Free$ preserves probabilistic transformations on $C$, such $\w'$ and $\w''$ exist.
 Note that $\alpha + \Tr \left[ (\id_A \ot \tGamma_0)(\w) \right] = 1$ holds
 since $\tGamma_1 + \tGamma_0$ is a channel.
 The inequality follows from $\Tr(e \w') \le \Tr(e \w^\opt)$ and Eq.~\eqref{eq:ExLocalNCP}.
 Therefore, from Eqs.~\eqref{eq:ExLocalGain1} and \eqref{eq:ExLocalGain2CP},
 Eq.~\eqref{eq:ExLocalGainResult} with $\varepsilon' = \varepsilon = 0$ holds.
\end{proof}

\begin{lemma} \label{lemma:ExLocalSym}
 Arbitrarily choose two systems $A$ and $C$, a natural number $N$, and
 a collection of subchannels $\{ \tLambda_n \}_{n=1}^N$ from $A$ to $C$.
 Define the collection of channels
 $\{ \Lambda_{i,q,r} \in \Chn_{A \to B \ot C} \}_{(i,q,r)=(1,1,1)}^{(N,N_C,N_C)}$ by
 \begin{alignat}{1}
  \Lambda_{i,q,r}(\sigma) &= (\tX^{i-1} \ot X^{q-1}Z^{r-1})
  \left[ \sum_{n=1}^N \ket{n}\bra{n} \ot \tLambda_n(\sigma) \right] \nonumber \\
  &\quad \cdot (\tX^{1-i} \ot Z^{1-r}X^{1-q}), \quad \forall \sigma \in \Den_A,
  \label{eq:ExLocalSym_Lambda}
 \end{alignat}
 where $B$ is an $N$-level system, $\{ \ket{n} \}_{n=1}^N$ is the computational basis of $B$,
 $\tX$ is the generalized Pauli-$X$ matrix of order $N$,
 and $X$ (resp. $Z$) is the generalized Pauli-$X$ (resp. Pauli-$Z$) matrix of order $N_C$.
 Also, let us consider the equal prior probabilities
 $\{ p_{i,q,r} \coloneqq 1 / (N N_C^2) \}_{(i,q,r)=(1,1,1)}^{(N,N_C,N_C)}$.
 Then, the following two properties hold.
 \begin{enumerate}[leftmargin=2em]
  \item For any measurement $\Pi \coloneqq \{ \Pi_{i,q,r} \}_{(i,q,r)=(1,1,1)}^{(N,N_C,N_C)}$
        on $B \ot C \ot C$, there exists a measurement
        $\Pi^\sym \coloneqq \{ \Pi^\sym_{i,q,r} \}_{(i,q,r)=(1,1,1)}^{(N,N_C,N_C)}$
        on $B \ot C \ot C$ satisfying
        \begin{alignat}{1}
         \hspace{1em}
         \PS(\rho,\Pi;\SetPLIQR) &= \PS(\rho,\Pi^\sym;\SetPLIQR) \nonumber \\
         &= \Tr \left[ \Pi^\sym_{1,1,1} \cdot (\Lambda_{1,1,1} \ot \id_C) (\rho) \right], \nonumber \\
         &\quad \forall \rho \in \Den_{A \ot C}
         \label{eq:ExLocalSymPS}
        \end{alignat}
        and
        \begin{alignat}{1}
         \Pi^\sym_{i,q,r} &= (\tX^{i-1} \ot X^{q-1}Z^{r-1} \ot I_{N_C}) \Pi^\sym_{1,1,1}
         \nonumber \\
         &\quad \cdot (\tX^{1-i} \ot Z^{1-r}X^{1-q} \ot I_{N_C}).
         \label{eq:ExLocalSymPi}
        \end{alignat}
  \item For such $\Pi^\sym$, let
        \begin{alignat}{1}
         \tGamma_n &\coloneqq N_C \left[ \id_C \ot \Tr \left[ \Pi^\sym_{1,1,1}
         \cdot (\ket{n}\bra{n} \ot \id_{C \ot C}) \right] \right] \nonumber \\
         &\quad \c (\Phi_C \ot \id_C) \in \CP_{C \to C}, \quad n \in \{ 1,2,\dots,N \};
         \label{eq:ExLocalSymGamma}
        \end{alignat}
        then, $\{ \tGamma_n \}_{n=1}^N$ are subchannels on $C$ that satisfy
        \begin{alignat}{1}
         &\hspace{-1em} \sum_{n=1}^N \Tr \left[ \Phi_C \cdot (\tLambda_n \ot \tGamma_n)(\rho) \right]
         \nonumber \\
         &= N_C \PS(\rho,\Pi^\sym;\SetPLIQR),
         \quad \forall \rho \in \Den_{A \ot C}.
         \label{eq:ExLocalSymGammaPS}
        \end{alignat}
 \end{enumerate}
\end{lemma}
\begin{proof}
 (1): Let us consider the following group:
 \begin{alignat}{1}
  \cG &\coloneqq \left\{ g_{i,q,r} \right\}_{(i,q,r)=(1,1,1)}^{(N,N_C,N_C)}, \nonumber \\
  g_{i,q,r} &\coloneqq (\tX^i \ot X^q Z^r \ot I_{N_C}) \cdot \Endash
  \cdot (\tX^{-i} \ot Z^{-r}X^{-q} \ot I_{N_C}) \nonumber \\
  &\quad \in \Chn_{B \ot C \ot C \to B \ot C \ot C}.
  \label{eq:ExLocalSymG}
 \end{alignat}
 Since $\Lambda_{\mod{i+i'}{N},\mod{q+q'}{N_C},\mod{r+r'}{N_C}} \ot \id_C
 = g_{i,q,r} \c (\Lambda_{i',q',r'} \ot \id_C)$ holds for any $i,i',q,q',r,r'$
 (where $\mod{n}{N}$ is defined in Eq.~\eqref{eq:mod}),
 the collection $\{ \Lambda_{i,q,r} \ot \id_C \}_{(i,q,r)=(1,1,1)}^{(N,N_C,N_C)}$
 is $\cG$-covariant.
 Note that the corresponding group $G$ is
 \begin{alignat}{1}
  G &\coloneqq \{ (i,q,r) \mid i \in \{1,2,\dots,N\}, ~q,r \in \{1,2,\dots,N_C\} \}
 \end{alignat}
 with the product $(i,q,r) \cdot (i',q',r')$ being
 $(\mod{i+i'}{N},\mod{q+q'}{N_C},\mod{r+r'}{N_C})$ and the identity element $(N,N_C,N_C)$.
 Therefore, from Lemma~\ref{lemma:Sym}, there exists a $\cG$-covariant measurement $\Pi^\sym$
 (i.e., satisfying Eq.~\eqref{eq:ExLocalSymPi}) that satisfies Eq.~\eqref{eq:ExLocalSymPS}.
 
 (2): We have
 \begin{alignat}{1}
  \Tr \c \left( \sum_{n=1}^N \tGamma_n \right)
  &= N_C \Tr \c \sum_{n=1}^N \left[ \id_C
  \ot \Tr \left[ \Pi^\sym_{1,1,1} \cdot (\ket{n}\bra{n} \ot \id_{C \ot C}) \right]
  \right] \nonumber \\
  &\quad \c (\Phi_C \ot \id_C) \nonumber \\
  &= N_C \Tr \c \left[ \id_C \ot \Tr \left[ \Pi^\sym_{1,1,1} \cdot (I_N \ot \id_{C \ot C}) \right] \right]
  \nonumber \\
  &\quad \c (\Phi_C \ot \id_C) \nonumber \\
  &= N_C \Tr \left[ \Pi^\sym_{1,1,1} \cdot (I_N \ot I_{N_C} \ot \id_C) \right] \nonumber \\
  &= N_C \Tr \left[ \left( \Tr_{B \ot C} \Pi^\sym_{1,1,1} \right) \cdot \Endash \right],
  \label{eq:ExLocalSymGammaSum}
 \end{alignat}
 where $\Tr_{B \ot C}$ is the partial trace over the first two systems $B$ and $C$ of
 the composite system $B \ot C \ot C$.
 For the channel $\cD$ on $B$ represented by the map
 $\cD \colon \Pos_B \ni \sigma \mapsto \sum_{n=1}^N \ket{n}\braket{n|\sigma|n}\bra{n} \in \Pos_B$,
 each $\Lambda_{i,q,r}$ satisfies $(\cD \ot \id_C) \c \Lambda_{i,q,r} = \Lambda_{i,q,r}$.
 Therefore, we assume, without loss of generality, that $\Pi$ satisfies
 $(\cD \ot \id_{C \ot C})(\Pi_{i,q,r}) = \Pi_{i,q,r}$ $~(\forall i,q,r)$,
 and thus $\Pi^\sym$ also satisfies
 $(\cD \ot \id_{C \ot C})(\Pi^\sym_{i,q,r}) = \Pi^\sym_{i,q,r}$ $~(\forall i,q,r)$.
 Since $\Pi^\sym$ is a measurement, we have
 \begin{alignat}{1}
  I_N \ot I_{N_C} \ot I_{N_C}
  &= \sum_{i=1}^N \sum_{q=1}^{N_C} \sum_{r=1}^{N_C} \Pi^\sym_{i,q,r} \nonumber \\
  &= \sum_{i=1}^N \sum_{q=1}^{N_C} \sum_{r=1}^{N_C} g_{i-1,q-1,r-1}(\Pi^\sym_{1,1,1}) \nonumber \\
  &= I_N \ot I_{N_C} \ot N_C \Tr_{B \ot C} \Pi^\sym_{1,1,1},
 \end{alignat}
 where the last equality follows from $\sum_{i=1}^N \sum_{q=1}^{N_C} \sum_{r=1}^{N_C} g_{i-1,q-1,r-1}
 = I_N \ot I_{N_C} \ot (N_C \Tr_{B \ot C} \Endash)$.
 This gives $N_C \Tr_{B \ot C} \Pi^\sym_{1,1,1} = I_{N_C}$, and thus
 $\Tr \c \left( \sum_{n=1}^N \tGamma_n \right) = \Tr \Endash$ holds
 from Eq.~\eqref{eq:ExLocalSymGammaSum}.
 Therefore, $\sum_{n=1}^N \tGamma_n$ is a channel, i.e., $\{ \tGamma_n \}_{n=1}^N$ are subchannels.
 We obtain
 \begin{alignat}{1}
  \hspace{1em} &\hspace{-1em}
  \sum_{n=1}^N \Tr \left[ \Phi_C \cdot (\tLambda_n \ot \tGamma_n)(\rho) \right] \nonumber \\
  &= N_C \sum_{n=1}^N \Tr \left[ \Phi_C \cdot \left[ \tLambda_n \ot
  \left[ \id_C \ot \Tr \left[ \Pi^\sym_{1,1,1} \cdot (\ket{n}\bra{n} \ot \id_{C \ot C}) \right]
  \right] \right. \right. \nonumber \\
  &\quad \left. \left. \vphantom{\tLambda_n} \c (\Phi_C \ot \id_C) \right] (\rho) \right] \nonumber \\
  &= N_C \sum_{n=1}^N \Tr \left[ \Phi_C \cdot
  \left[ \left[ \id_C \ot \Tr \left[ \Pi^\sym_{1,1,1} \cdot (\ket{n}\bra{n} \ot \id_{C \ot C}) \right]
  \right] \right. \right. \nonumber \\
  &\quad \left. \left. \vphantom{\tLambda_n} \c (\tLambda_n \ot \Phi_C \ot \id_C) \right] (\rho)
  \right] \nonumber \\
  &= N_C \sum_{n=1}^N \Tr \left[ \Pi^\sym_{1,1,1} \cdot \left[ (\ket{n}\bra{n} \ot \id_{C \ot C})
  \c (\tLambda_n \ot \id_C) \right] (\rho) \right] \nonumber \\
  &= N_C \sum_{n=1}^N \Tr \left[ \Pi^\sym_{1,1,1} \cdot (\ket{n}\bra{n} \ot \tLambda_n \ot \id_C)
  (\rho) \right] \nonumber \\
  &= N_C \Tr \left[ \Pi^\sym_{1,1,1} \cdot (\Lambda_{1,1,1} \ot \id_C)
  (\rho) \right] \nonumber \\
  &= N_C \PS(\rho,\Pi^\sym;\SetPLIQR),
 \end{alignat}
 and thus Eq.~\eqref{eq:ExLocalSymGammaPS} holds.
 The last equality follows from the right-hand equality of Eq.~\eqref{eq:ExLocalSymPS}.
\end{proof}

\begin{proof}[Theorem~\ref{thm:ExLocal}]
 Since $G_C(\rho,\Free) \le 1 + R_\Free(\rho)$ holds as stated in the main text,
 it suffices to show $G_C(\rho,\Free) \ge 1 + R_\Free(\rho)$.
 Choose an arbitrary positive real number $\varepsilon$, and consider a collection of subchannels
 $\{ \tLambda_n \}_{n=1}^N$ from $A$ to $C$ obtained by Lemma~\ref{lemma:ExLocalGain}.
 Define the collection of channels
 $\{ \Lambda_{i,q,r} \in \Chn_{A \to B \ot C} \}_{(i,q,r)=(1,1,1)}^{(N,N_C,N_C)}$
 by Eq.~\eqref{eq:ExLocalSym_Lambda}, and consider the equal prior probabilities
 $\{ p_{i,q,r} \coloneqq 1 / (N N_C^2) \}_{(i,q,r)=(1,1,1)}^{(N,N_C,N_C)}$.
 Furthermore, define the measurement
 $\Psi \coloneqq \{ \Psi_{i,q,r} \}_{(i,q,r)=(1,1,1)}^{(N,N_C,N_C)}$ on $B \ot C \ot C$ by
 \begin{alignat}{1}
  \Psi_{i,q,r} &\coloneqq 
  (\tX^{i-1} \ot X^{q-1}Z^{r-1} \ot I_{N_C}) (\ket{1}\bra{1} \ot N_C^{-1} \Phi_C) \nonumber \\
  &\quad \cdot (\tX^{1-i} \ot Z^{1-r}X^{1-q} \ot I_{N_C}).
  \label{eq:ExLocalPsi}
 \end{alignat}
 Define $\cG$ by Eq.~\eqref{eq:ExLocalSymG}; then, both
 $\{ p_{i,q,r}, \Lambda_{i,q,r} \ot \id_C \}$ and $\Psi$ are $\cG$-covariant.
 Let us consider $\tGamma_n$ obtained by substituting
 $\Psi_{1,1,1} = \ket{1}\bra{1} \ot N_C^{-1} \Phi_C$ into $\Pi^\sym_{1,1,1}$
 in Eq.~\eqref{eq:ExLocalSymGamma}.
 Then, from
 \begin{alignat}{1}
  \tGamma_n &= \left[ \id_C \ot \Tr[(\ket{1}\bra{1} \ot \Phi_C) \cdot (\ket{n}\bra{n} \ot \id_{C \ot C})]
  \right] \c (\Phi_C \ot \id_C) \nonumber \\
  &= |{\braket{1|n}}|^2 \cdot \id_C,
 \end{alignat}
 $\tGamma_1 = \id_C$ and $\tGamma_2 = \tGamma_3 = \dots = \tGamma_N = \zero$ hold,
 and thus we obtain
 $\Tr \left[ \Phi_C \cdot (\tLambda_1 \ot \id_C)(\rho) \right] = N_C \PS(\rho,\Psi;\SetPLIQR)$
 from Eq.~\eqref{eq:ExLocalSymGammaPS} with $\Pi^\sym$ substituted by $\Psi$.
 Therefore, we have
 \begin{alignat}{1}
  \PS(\rho;\SetPLIQR) &\ge \PS(\rho,\Psi;\SetPLIQR) \nonumber \\
  &= N_C^{-1} \Tr \left[ \Phi_C \cdot (\tLambda_1 \ot \id_C)(\rho) \right].
  \label{eq:ExLocalMain1}
 \end{alignat}

 Let $\w^\opt$ be $\w \in \cl \Free$ that maximizes $\PS(\w;\SetPLIQR)$.
 From Lemma~\ref{lemma:ExLocalSym}, there exists a $\cG$-covariant measurement $\Pi^\sym$
 that satisfies $\PS(\w^\opt;\SetPLIQR) = \PS(\w^\opt,\Pi^\sym;\SetPLIQR)$.
 Define the collection of subchannels $\{ \tGamma_n \}_{n=1}^N$ by Eq.~\eqref{eq:ExLocalSymGamma};
 then, we have
 \begin{alignat}{1}
  \sup_{\w \in \Free} \PS(\w;\SetPLIQR) &= \PS(\w^\opt,\Pi^\sym;\SetPLIQR) \nonumber \\
  &= N_C^{-1} \sum_{n=1}^N \Tr \left[ \Phi_C \cdot (\tLambda_n \ot \tGamma_n)(\w^\opt) \right],
  \nonumber \\
  \label{eq:ExLocalMain2}
 \end{alignat}
 where the right-hand equality follows from Eq.~\eqref{eq:ExLocalSymGammaPS}.
 Thus, from Eqs.~\eqref{eq:ExLocalMain1} and \eqref{eq:ExLocalMain2}, we obtain
 \begin{alignat}{1}
  &\hspace{-1em} \frac{\PS(\rho;\SetPLIQR)}{\sup_{\w \in \Free} \PS(\w;\SetPLIQR)} \nonumber \\
  &\ge \frac{\Tr \left[ \Phi_C \cdot (\tLambda_1 \ot \id_C)(\rho) \right]}{
  \sum_{n=1}^N \Tr \left[ \Phi_C \cdot (\tLambda_n \ot \tGamma_n)(\w^\opt) \right]}
  \nonumber \\
  &\ge \frac{\Tr \left[ \Phi_C \cdot (\tLambda_1 \ot \id_C)(\rho) \right]}{
  \sup_{\w \in \Free}\sum_{n=1}^N \Tr \left[ \Phi_C \cdot (\tLambda_n \ot \tGamma_n)(\w) \right]}
  \nonumber \\
  &\ge \frac{1 + R_\Free(\rho)}{1 + \varepsilon},
  \label{eq:ExLocalMain}
 \end{alignat}
 where the second inequality follows from
 \begin{alignat}{1}
  \sum_{n=1}^N \Tr \left[ \Phi_C \cdot (\tLambda_n \ot \tGamma_n)(\w^\opt) \right]
  &\le \max_{\w \in \cl \Free}\sum_{n=1}^N \Tr \left[ \Phi_C \cdot (\tLambda_n \ot \tGamma_n)(\w) \right]
  \nonumber \\
  &= \sup_{\w \in \Free}\sum_{n=1}^N \Tr \left[ \Phi_C \cdot (\tLambda_n \ot \tGamma_n)(\w) \right]
 \end{alignat}
 and the last inequality follows from Lemma~\ref{lemma:ExLocalGain}.
 Since this equation holds for any positive real number $\varepsilon$,
 we have $G_C(\rho,\Free) \ge 1 + R_\Free(\rho)$.
 Therefore, $G_C(\rho,\Free) = 1 + R_\Free(\rho)$ holds.
\end{proof}

In particular, if $\Free$ preserves probabilistic transformations on $C$, then
Eq.~\eqref{eq:ExLocalMain} with $\varepsilon = 0$ holds,
and thus the supremum $\sup_{\SetPLIQR}$ of $G_C(\rho,\Free)$ can be replaced by
the maximum $\max_{\SetPLIQR}$.
We provide two typical examples of subsets of $\Den_{A \ot C}$ that preserve
probabilistic transformations on $C$.
The first example is the set of states with Schmidt number at most $k$ (where $k$ is any natural number)
\cite{Ter-Hor-2000}.
The second example is the set of all unsteerable states from $A$ to $C$ \cite{Gal-Aol-2015}.
Note that the set of all unsteerable states from $C$ to $A$ does not necessarily preserve
probabilistic transformations on $C$ \cite{Pra-Cho-Han-Lee-2019}.

\section{The case in which positive definite free states do not exist} \label{sec:EmptyInterior}

\subsection{Without considering auxiliary systems}

Let us consider the set of free states $\Free$ of an arbitrary system $A$.
Assume that $\Free$ is not empty so that $G(\rho,\Free)$ and $R_\Free(\rho)$ are well-defined
for at least one state $\rho \in \Den_A$.
For any subset $T$ of $\Den_A$, let $S_T$ be the set of all $\rho \in \Den_A$ such that
there exists $\w \in T$ and a positive real number $c$ satisfying $c \rho \le \w$.
$S_T = \Den_A$ is equivalent to $T$ containing a positive definite matrix.

The following lemma holds.
\begin{lemma} \label{lemma:SFree}
 For any $\rho \in \Den_A$, (1) $\rho \in S_{\co \Free}$,
 (2) $R_{\co \Free}(\rho) < \infty$, and (3) $G(\rho,\Free) < \infty$ are all equivalent.
\end{lemma}
\begin{proof}
 $(1) \Rightarrow (2)$:
 Assume $\rho \in S_{\co \Free}$; then, there exists $\w \in \co \Free$ and
 a positive real number $c$ such that $c \rho \le \w$ holds.
 $c \le 1$ clearly holds, and if $c = 1$ holds, then $\rho = \w$ implies $R_{\co \Free}(\rho) = 0$.
 If $c < 1$ holds, then $\lambda \coloneqq \frac{1}{c} - 1$ and
 $\tau \coloneqq \frac{\w - c \rho}{1 - c}$ satisfy $\tau \in \Den_A$ and
 \begin{alignat}{1}
  \frac{\rho + \lambda \tau}{1 + \lambda} &= c \rho + (1 - c) \tau = c \rho + (\w - c \rho) = \w.
 \end{alignat}
 Thus, we have $R_{\co \Free}(\rho) \le \lambda < \infty$.

 $(2) \Rightarrow (3)$:
 $G(\rho,\Free) = G(\rho,\co \Free) \le 1 + R_{\co \Free}(\rho) < \infty$ obviously holds,
 where the equality follows from the fact
 the proof of $G(\rho,\Free) = G(\rho,\co \Free)$ in Corollary~\ref{cor:Main} applies directly,
 and the first inequality follows from the fact the proof of $G(\rho,\Free) \le 1 + R_\Free(\rho)$
 in Theorem~\ref{thm:MainConvex} with $\Free$ replaced by $\co \Free$ applies directly.

 $(3) \Rightarrow (1)$:
 We show the contrapositive, i.e., $G(\rho,\Free) = \infty$ for any $\rho \not\in S_{\co \Free}$.
 Let $\V$ be the smallest complex vector space containing the support spaces of
 all $\w \in \co \Free$; then,
 $S_{\co \Free} = \{ \sigma \in \Den_A \mid \supp \sigma \subseteq \V \}$ holds.
 Let $P \in \Pos_A$ be the projection matrix onto the orthogonal complement of $\V$.
 $\Tr(P \w) = 0$ holds for any $\w \in \Free$ since $\supp \w \subseteq \V$ holds.
 Also, $\Tr(P \rho) > 0$ holds from $\supp \rho \not\subseteq \V$.
 Choose any natural number $N$ and an $N$-level system $B$.
 Let $p_n \coloneqq 1/N$ $~(\forall n \in \{1,2,\dots,N\})$ and
 define channels $\Lambda_1,\dots,\Lambda_N$ from $A$ to $B$ by
 \begin{alignat}{1}
  \Lambda_n(\sigma) &= \Tr(P \sigma) \ket{n}\bra{n}
  + \frac{1 - \Tr(P \sigma)}{N-1} (I_N - \ket{n}\bra{n}), \nonumber \\
  &\quad \sigma \in \Den_A.
 \end{alignat}
 Then, we can use a similar method as in the proof of Theorem~\ref{thm:MainConvex}.
 Specifically, let $\cG \coloneqq \{ X^k (\Endash) X^{-k} \}_{k=1}^N$,
 where $X$ is the generalized Pauli-$X$ matrix of order $N$; then, $\SetL$ is $\cG$-covariant.
 Thus, from Lemma~\ref{lemma:Sym}, for each measurement $\Pi$, there exists
 a $\cG$-covariant measurement $\Pi^\sym$ satisfying
 \begin{alignat}{1}
  \PS(\sigma,\Pi;\SetPL) &= \Tr \left[ \Pi^\sym_1 \cdot \Lambda_1(\sigma) \right]
  \nonumber \\
  &= c \Tr(P \sigma) + (1 - c) \frac{1 - \Tr(P \sigma)}{N-1}, \nonumber \\
  &\quad \forall \sigma \in \Den_A,
  \label{eq:EmptyInterior_PS_qc}
 \end{alignat}
 where $c \coloneqq \braket{1|\Pi^\sym_1|1}$.
 The right-hand equality follows from $\Tr \Pi^\sym_1 = 1$, which is derived from
 $\sum_{n=1}^N \Pi^\sym_n = I_N$.
 Therefore, we have, for each $\w \in \Free$,
 \begin{alignat}{1}
  \PS(\w,\Pi;\SetPL) &= c \Tr(P \w) + (1 - c) \frac{1 - \Tr(P \w)}{N-1} \nonumber \\
  &= \frac{1-c}{N-1} \le \frac{1}{N-1}.
 \end{alignat}
 For a measurement $\Pi$ satisfying $c = \braket{1|\Pi^\sym_1|1} = 1$,
 since $\PS(\rho,\Pi;\SetPL) = \Tr(P \rho)$ holds from Eq.~\eqref{eq:EmptyInterior_PS_qc},
 we obtain
 \begin{alignat}{1}
  G(\rho,\Free) &\ge \frac{\PS(\rho;\SetPL)}{\sup_{\w \in \Free} \PS(\w;\SetPL)}
  \ge \frac{\Tr(P \rho)}{\frac{1}{N-1}} = (N-1) \Tr(P \rho).
 \end{alignat}
 Thus, taking the limit as $N \to \infty$ yields $G(\rho,\Free) = \infty$.
\end{proof}

It is easy to see that Theorem~\ref{thm:MainConvex} and Corollary~\ref{cor:Main} hold
when $R_{\co \Free}(\rho) < \infty$, or equivalently $G(\rho,\Free) < \infty$, holds.
That is, even if $\co \Free$ does not contain a positive definite matrix,
Theorem~\ref{thm:MainConvex} and Corollary~\ref{cor:Main} hold for any $\rho \in S_{\co \Free}$,
and, from Lemma~\ref{lemma:SFree}, $R_{\co \Free}(\rho) = \infty$
and $G(\rho,\Free) = \infty$ hold for any $\rho \not\in S_{\co \Free}$.

\subsection{Considering auxiliary systems}

For any two systems $A$ and $C$, let us consider the set of free states $\Free$
of the composite system $A \ot C$.
Assume that $\Free$ is not empty; then, the following corollary holds.
\begin{cor} \label{cor:SFreeEx}
 For any $\rho \in \Den_{A \ot C}$, (1) $\rho \in S_{\co \Free_C}$,
 (2) $R_{\co \Free_C}(\rho) < \infty$, and (3) $G_C(\rho,\Free_C) < \infty$ are all equivalent.
\end{cor}
\begin{proof}
 Substituting $\Free_C$ for $\Free$ in Lemma~\ref{lemma:SFree} gives $(1) \Leftrightarrow (2)$.
 If $R_{\co \Free_C}(\rho) < \infty$ holds, then
 we have $G_C(\rho,\Free_C) \le G(\rho,\Free_C) < \infty$,
 where the right-hand inequality follows from $(2) \Rightarrow (3)$
 from Lemma~\ref{lemma:SFree} with $\Free$ replaced by $\Free_C$.
 Thus, $(2) \Rightarrow (3)$ holds.
 It remains to show $(3) \Rightarrow (1)$.
 
 We show the contrapositive of $(3) \Rightarrow (1)$, i.e., $G_C(\rho,\Free_C) = \infty$,
 or equivalently $G_C(\rho,\co \Free_C) = \infty$, for any $\rho \not\in S_{\co \Free_C}$.
 From the contrapositive of $(1) \Leftarrow (2)$, we have $R_{\co \Free_C}(\rho) = \infty$.
 Choose an arbitrary positive real number $R$, and a feasible solution $x$
 of the problem with $\co \Free_C$ substituted for $\Free$ in Problem~\ref{prob:R_dual}
 that satisfies $\Tr(x \rho) - 1 \ge R$, i.e., $x \in \Pos_{A \ot C}$ satisfying
 \begin{alignat}{1}
  \Tr(x \rho) \ge 1 + R, \quad \Tr(x \w) \le 1 ~(\forall \w \in \co \Free_C).
 \end{alignat}
 Since the optimal value of this problem is $R_{\co \Free_C}(\rho) = \infty$,
 such $x$ exists.
 Fix a positive real number $\varepsilon$; then, similarly to Lemma~\ref{lemma:ExLocalGain},
 there exist a natural number $N$ and
 a collection of subchannels $\{ \tLambda_n \}_{n=1}^N$ from $A$ to $C$ such that
 \begin{alignat}{1}
  \frac{\Tr \left[ \Phi_C \cdot (\tLambda_1 \ot \id_C)(\rho) \right]}{
  \sup_{\w \in \co \Free_C} \sum_{n=1}^N \Tr \left[ \Phi_C \cdot (\tLambda_n \ot \tGamma_n)(\w) \right]}
  &\ge \frac{1 + R}{1 + \varepsilon}
  \label{eq:ExLocalGainSFree}
 \end{alignat}
 holds for any subchannels $\{ \tGamma_n \}_{n=1}^N$ on $C$.
 Indeed, we can use $x$ instead of $x^\opt$ in the proof of Lemma~\ref{lemma:ExLocalGain}
 (note that $\co \Free_C$ preserves operations on $C$).
 The rest can be proved similarly to the proof of Theorem~\ref{thm:ExLocal}.
 Define the collection of channels
 $\{ \Lambda_{i,q,r} \in \Chn_{A \to B \ot C} \}_{(i,q,r)=(1,1,1)}^{(N,N_C,N_C)}$
 by Eq.~\eqref{eq:ExLocalSym_Lambda}, and consider the equal prior probabilities
 $\{ p_{i,q,r} \coloneqq 1 / (N N_C^2) \}_{(i,q,r)=(1,1,1)}^{(N,N_C,N_C)}$.
 Also, let us consider the measurement
 $\Psi \coloneqq \{ \Psi_{i,q,r} \}_{(i,q,r)=(1,1,1)}^{(N,N_C,N_C)}$ on $B \ot C \ot C$
 given by Eq.~\eqref{eq:ExLocalPsi}.
 Then, similarly to Eq.~\eqref{eq:ExLocalMain1}, we have
 \begin{alignat}{1}
  \PS(\rho;\SetPLIQR) &\ge N_C^{-1} \Tr \left[ \Phi_C \cdot (\tLambda_1 \ot \id_C)(\rho) \right].
  \label{eq:ExLocalMain1_SFree}
 \end{alignat}
 Let $\w^\opt$ be $\w \in \cl (\co \Free_C)$ that maximizes $\PS(\w;\SetPLIQR)$.
 For the group $\cG$ given by Eq.~\eqref{eq:ExLocalSymG}, from Lemma~\ref{lemma:ExLocalSym},
 there exists a $\cG$-covariant measurement $\Pi^\sym$ that satisfies
 $\PS(\w^\opt;\SetPLIQR) = \PS(\w^\opt,\Pi^\sym;\SetPLIQR)$.
 Define the collection of subchannels $\{ \tGamma_n \}_{n=1}^N$ by Eq.~\eqref{eq:ExLocalSymGamma};
 then, similarly to Eq.~\eqref{eq:ExLocalMain2}, we have
 \begin{alignat}{1}
  \sup_{\w \in \co \Free_C} \PS(\w;\SetPLIQR)
  &= N_C^{-1} \sum_{n=1}^N \Tr \left[ \Phi_C \cdot (\tLambda_n \ot \tGamma_n)(\w^\opt) \right].
  \label{eq:ExLocalMain2_SFree}
 \end{alignat}
 Therefore, from Eqs.~\eqref{eq:ExLocalMain1_SFree} and \eqref{eq:ExLocalMain2_SFree}, we obtain
 \begin{alignat}{1}
  &\hspace{-1em} \frac{\PS(\rho;\SetPLIQR)}{\sup_{\w \in \co \Free_C} \PS(\w;\SetPLIQR)} \nonumber \\
  &\ge \frac{\Tr \left[ \Phi_C \cdot (\tLambda_1 \ot \id_C)(\rho) \right]}{
  \sum_{n=1}^N \Tr \left[ \Phi_C \cdot (\tLambda_n \ot \tGamma_n)(\w^\opt) \right]}
  \nonumber \\
  &\ge \frac{\Tr \left[ \Phi_C \cdot (\tLambda_1 \ot \id_C)(\rho) \right]}{
  \sup_{\w \in \co \Free_C}\sum_{n=1}^N \Tr \left[ \Phi_C \cdot (\tLambda_n \ot \tGamma_n)(\w) \right]}
  \nonumber \\
  &\ge \frac{1 + R}{1 + \varepsilon},
 \end{alignat}
 where the last inequality follows from Eq.~\eqref{eq:ExLocalGainSFree}.
 Taking the limit as $R \to \infty$, we obtain $G_C(\rho,\co \Free_C) = \infty$.
\end{proof}

Similarly to Theorem~\ref{thm:MainConvex} and Corollary~\ref{cor:Main},
it is easy to see that Theorem~\ref{thm:ExLocal} and Corollary~\ref{cor:ExLocalGeneral} hold
when $R_{\co \Free_C}(\rho) < \infty$, or equivalently $G_C(\rho,\Free_C) < \infty$, holds.
That is, even if $\co \Free_C$ does not contain a positive definite matrix,
Theorem~\ref{thm:ExLocal} and Corollary~\ref{cor:ExLocalGeneral} hold
for any $\rho \in S_{\co \Free_C}$,
and, from Corollary~\ref{cor:SFreeEx}, $R_{\co \Free_C}(\rho) = \infty$
and $G_C(\rho,\Free_C) = \infty$ hold for any $\rho \not\in S_{\co \Free_C}$.

\section{The measure defined by the left-hand side of Eq.~(\ref{eq:H0})} \label{sec:GvsH}

Instead of the discrimination power, i.e.,
\begin{alignat}{1}
 G(\rho,\Free) &= \sup_{\SetPL}
 \frac{\PS(\rho;\SetPL)}{\sup_{\w \in \Free} \PS(\w;\SetPL)},
 \label{eq:GH_G}
\end{alignat}
the measure given by the left-hand side of Eq.~\eqref{eq:H0}, i.e.,
\begin{alignat}{1}
 H(\rho,\Free) &\coloneqq \sup_{\SetPL} \max_\Pi
 \frac{\PS(\rho,\Pi;\SetPL)}{\sup_{\w \in \Free} \PS(\w,\Pi;\SetPL)},
 \label{eq:H}
\end{alignat}
has been used in the literature \cite{Tak-Reg-Bu-Liu-2019,Uol-Kra-Sha-Yu-2019,Reg-Lam-Fer-Tak-2021,
Kur-Tak-Ade-Yam-2024,Tur-Gut-Ade-2024,Lam-Reg-Tak-Fer-2021}.
While $G(\rho,\Free)$ and $H(\rho,\Free)$ might appear similar in form,
they differ significantly in two aspects:
(a) intuitive operational meaning and (b) quantification.
We will discuss each aspect in detail below.
Note that the same arguments hold if we consider $G_C(\rho,\Free)$ instead of $G(\rho,\Free)$
and extend $H(\rho,\Free)$ to use an auxiliary system $C$.

As a preliminary step, we present the following lemma.
\begin{lemma} \label{lemma:H}
 Let $E$ be the set of all $e \in \Pos_A$ satisfying $\|e\|_\infty \le 1$
 (often referred to as effects on system $A$).
 Then, we have
 \begin{alignat}{1}
  H(\rho,\Free) &= \sup_{e \in E} \frac{\Tr(e \rho)}{\sup_{\w \in \Free} \Tr(e \w)}.
  \label{eq:H_eq_R}
 \end{alignat}
\end{lemma}
\begin{proof}
 We have, for any $\sigma \in \Den_A$,
 \begin{alignat}{1}
  \PS(\sigma,\Pi;\SetPL) &= \sum_{n=1}^N p_n \Tr[\Pi_n \cdot \Lambda_n(\sigma)]
  = \Tr(e \sigma),
 \end{alignat}
 where
 \begin{alignat}{1}
  e &\coloneqq \sum_{n=1}^N p_n \cdot \Lambda_n^\dagger(\Pi_n).
 \end{alignat}
 Also, $e$ is in $E$ since $0 \le \Tr(e \sigma) \le 1 ~(\forall \sigma \in \Den_A)$ holds.
 Thus, from the definition of $H(\rho,\Free)$ (i.e., Eq.~\eqref{eq:H}),
 the left-hand side of Eq.~\eqref{eq:H_eq_R} is greater than or equal to the right-hand side.
 Conversely, for any $e \in E$, considering the case in which $p_1 = 1$, $\Lambda_1 = \id_A$,
 and $\Pi_1 = e$, we obtain $\PS(\sigma,\Pi;\SetPL) = \Tr(e \sigma) ~(\forall \sigma \in \Den_A)$
 from Eq.~\eqref{eq:PS_rho_Pi},
 so the left-hand side of Eq.~\eqref{eq:H_eq_R} is less than or equal to the right-hand side.
 Therefore, Eq.~\eqref{eq:H_eq_R} holds.
\end{proof}

First, we discuss (a).
$G(\rho,\Free)$ is the upper bound on the ratio of the maximum average success probabilities
in channel discrimination problems when using the state $\rho$ compared to using the best free state.
In this way, $G(\rho,\Free)$ can be straightforwardly explained in operational terms.
Generally, a channel discrimination problem involves finding the maximum average success probability
given a set of prior probabilities and channels $\SetPL$.
Since the discrimination power focuses on the ratio of maximum average success probabilities,
it is intuitively understandable.
On the other hand, $H(\rho,\Free)$ considers the ratio of average success probabilities
when the same measurement $\Pi$ is used for discrimination,
comparing the given state to the best free state, and takes the supremum over $\SetPL$ and $\Pi$.
Understanding the meaning of this ratio in the context of channel discrimination problems
may not be straightforward.
Since Eq.~\eqref{eq:H_eq_R} can be easily derived, it might be more straightforward to
interpret $H(\rho,\Free)$ using the probabilities obtained when applying effects.

Next, we discuss (b).
It is not immediately obvious that Theorems~\ref{thm:MainConvex} and \ref{thm:ExLocal} hold.
In fact, the results corresponding to these theorems in specific resource theories have been derived
in Refs.~\cite{Bae-Chr-Pia-2019,Nap-Bro-Cia-Pia-2016,Bu-Sin-Fei-Pat-2017,Pia-Cia-Bro-Nap-2016,
Tak-Reg-Bu-Liu-2019}, but they have not been shown in a form applicable to
arbitrary resource theories.
On the other hand, it is easy to prove $H(\rho, \Free) = 1 + R_\Free(\rho)$
(which gives $H(\rho, \Free) = G(\rho, \Free)$ from Theorem~\ref{thm:MainConvex})
in convex resource theories.
Indeed, letting $x \coloneqq e / \sup_{\w \in \Free} \Tr(e \w)$
for each $e$ on the right-hand side of Eq.~\eqref{eq:H_eq_R} yields
\begin{alignat}{1}
 H(\rho,\Free) &= \sup_{x \in E'} \Tr(x \rho),
 \label{eq:H_eq_xrho}
\end{alignat}
where $E' \coloneqq \{ x \in \Pos_A \mid \Tr(x \w) \le 1 ~(\forall \w \in \Free)\}$.
The right-hand side is 1 plus the optimal value of Problem~\ref{prob:R_dual},
i.e., $1 + R_\Free(\rho)$, so $H(\rho, \Free) = 1 + R_\Free(\rho)$ holds.
Note that a similar proof can be found in Theorem~1 of Ref.~\cite{Tak-Reg-2019}.

To summarize, $H(\rho, \Free)$ is a measure that can be easily understood through
Lemma~\ref{lemma:H} as $H(\rho, \Free) = 1 + R_\Free(\rho)$ (if $\Free$ is convex).
However, when attempting to operationally interpret $H(\rho, \Free)$ based on
the channel discrimination problem, it may not be straightforward to understand.
Compared to $H(\rho, \Free)$, $G(\rho, \Free)$ (or $G_C(\rho, \Free)$) is easier to
operationally interpret based on the channel discrimination problem,
but proving Theorem~\ref{thm:MainConvex} (or Theorem~\ref{thm:ExLocal}) for them
is not as straightforward.

%


\end{document}